\documentclass[aps,prx,twocolumn,amsmath,amssymb,nofootinbib,superscriptaddress]{revtex4-2}
\usepackage{times}
\usepackage{braket}
\usepackage{amsmath}
\usepackage{booktabs}
\usepackage{tabularx}

\newtheorem{proposition}{Proposition}
\usepackage{mathtools} 
\usepackage[colorlinks,linkcolor=blue,anchorcolor=blue,citecolor=blue]{hyperref}
\usepackage{multirow}
\usepackage{makecell}

\begin{document}
\title{Routing Codes: High-Rate Quantum LDPC Codes with Short, Parallel Non-Local Connectivity}

\author{Jiaxuan Zhang}
\email{zhangjx1@iai.ustc.edu.cn}
\affiliation{Institute of Artificial Intelligence, Hefei Comprehensive National Science Center, Hefei, Anhui, 230088, China}

\author{Zhao-Yun Chen}
\email{chenzhaoyun@iai.ustc.edu.cn}
\affiliation{Institute of Artificial Intelligence, Hefei Comprehensive National Science Center, Hefei, Anhui, 230088, China}

\author{Peng Duan}
\affiliation{Laboratory of Quantum Information, University of Science and Technology of China, Hefei, 230026, China}
\affiliation{Anhui Province Key Laboratory of Quantum Network, University of Science and Technology of China, Hefei, 230026, China}

\author{Jia-Ning Li}
\affiliation{Laboratory of Quantum Information, University of Science and Technology of China, Hefei, 230026, China}
\affiliation{Anhui Province Key Laboratory of Quantum Network, University of Science and Technology of China, Hefei, 230026, China}

\author{Tian-Hao Wei}
\affiliation{Laboratory of Quantum Information, University of Science and Technology of China, Hefei, 230026, China}
\affiliation{Anhui Province Key Laboratory of Quantum Network, University of Science and Technology of China, Hefei, 230026, China}

\author{Qing-Yang Hou}
\affiliation{Laboratory of Quantum Information, University of Science and Technology of China, Hefei, 230026, China}
\affiliation{Anhui Province Key Laboratory of Quantum Network, University of Science and Technology of China, Hefei, 230026, China}

\author{Wei-Cheng Kong}
\affiliation{Origin Quantum Computing,  Hefei, China}

\author{Yu-Chun Wu}
\email{wuyuchun@ustc.edu.cn}
\affiliation{Laboratory of Quantum Information, University of Science and Technology of China, Hefei, 230026, China}
\affiliation{Anhui Province Key Laboratory of Quantum Network, University of Science and Technology of China, Hefei, 230026, China}
\affiliation{Institute of Artificial Intelligence, Hefei Comprehensive National Science Center, Hefei, Anhui, 230088, China}

\author{Guo-Ping Guo}
\affiliation{Laboratory of Quantum Information, University of Science and Technology of China, Hefei, 230026, China}
\affiliation{Anhui Province Key Laboratory of Quantum Network, University of Science and Technology of China, Hefei, 230026, China}
\affiliation{Institute of Artificial Intelligence, Hefei Comprehensive National Science Center, Hefei, Anhui, 230088, China}
\affiliation{Origin Quantum Computing,  Hefei, China}

\date{\today}

\begin{abstract}
Quantum low-density parity-check (qLDPC) codes are promising candidates for realizing large-scale fault-tolerant quantum computing. Although many codes with favorable theoretical parameters have been developed, their practical adoption must take hardware implementability into account. For mainstream quantum platforms such as superconductors and neutral atoms, the connectivity, the length of non-local couplings, and the complexity of wiring or atom rearrangement are key factors that dictate the difficulty of hardware realization. Here, we propose a new family of qLDPC codes, termed routing codes. Within this family, we find explicit instances whose encoding rates are competitive to those of bivariate bicycle (BB) codes, while systematically reducing qubit connectivity, shortening the length of non-local couplings, and, crucially, making all non-local couplings mutually parallel. This parallelism translates into quantifiable benefits, substantially reducing wiring crossings in superconducting multi-layer architectures and simplifying the scheduling of atom movement in neutral-atom arrays. The weight-7 routing codes maintain a high threshold of 0.7\% and reduce the physical qubit overhead by approximately a factor of 8, compared to surface codes achieving a same logical error rate. Further increasing the weight to 11 yields a [[200, 24, 14]] and [[200, 16, 17]] codes with an encoding rate over 23 times that of the surface code, at the cost of a higher logical error rate. These results establish routing codes as a hardware-centric qLDPC family that bridges the gap between theoretical optimality and near-term physical feasibility. 

\end{abstract}
\maketitle

\section{Introduction}
Quantum computing promises transformative capabilities for classically intractable problems~\cite{365700,freedman2002simulation}, yet its physical realization is fundamentally limited by errors. Quantum error correction (QEC) provides a pathway to fault-tolerant computation by redundantly encoding logical qubits into a larger number of physical qubits~\cite{preskill1998reliable,gottesman1997stabilizer}. For over a decade, the surface code~\cite{PhysRevA.86.032324} has been the predominant quantum error-correcting architecture, primarily due to its high threshold and hardware-friendly nearest-neighbor geometry on a two-dimensional lattice. However, the most critical drawback of surface code is an extremely low encoding rate, which imposes a prohibitive qubit overhead for the practical large-scale computation. 

This limitation has recently spurred intense interest in quantum low-density parity-check (qLDPC) codes~\cite{PRXQuantum.2.040101}, which offer much higher encoding rates. Achieving higher encoding rates than local codes, however, requires long-range connectivity. Rigorous results show that any qLDPC code surpassing the surface code in both rate and distance must incorporate a growing amount of non-local couplings~\cite{Baspin2022connectivity,PhysRevLett.129.050505}. Consequently, a fundamental reality in the field is that the majority of theoretical studies on qLDPC codes to date assume an idealized, all-to-all connected hardware.

However, such all-to-all connectivity comes at a cost on any quantum platform. For neutral atoms and trapped ions, atom rearrangement or ion shuttling introduces decoherence, heating, and qubit loss that scale directly with the length of non-local couplings~\cite{Bluvstein2022,Bruzewicz2019}. On superconducting platforms, the challenges are even more severe. Higher connectivity demands complex wiring that exacerbates crosstalk, frequency crowding, and leakage errors~\cite{PhysRevApplied.10.054062, kosen2024signal}, while non-parallel long-range couplings cause wiring trajectories to cross. Resolving these crossings with multi-layer architectures forces detours that can inflate the actual wiring length by several times, introducing parasitic capacitance and reducing fabrication yield~\cite{Rosenberg2017}.

A central challenge, therefore, is to design qLDPC codes that preserve high encoding rates while fundamentally simplifying their hardware requirements. A landmark advance in this direction is the bivariate bicycle (BB) code~\cite{bravyi2024high}, which achieves high encoding rates and high threshold with a thickness-2 structure. However, the standard BB code such as the gross code requires degree-6 qubit connectivity and non-local coupling vectors of $(6,3)$ and $(3,6)$, whose crossing nature leads to severe wiring congestion in superconducting layouts and complex atom-movement scheduling in neutral-atom arrays. Although several circuit optimization techniques~\cite{PhysRevLett.134.090602,Zhou2026} have been proposed that moderately lower the required connectivity, the underlying pattern of non-local coupling vectors and the wiring crossings they cause remains unchanged. In contrast, a novel direction is to design new qLDPC codes that preserve the encoding rate advantages of BB codes while fundamentally simplifying the underlying connectivity structure.

In this work, we propose \textit{routing codes}, a family of quantum error-correcting codes constructed by qubit routing techniques~\cite{c31l-qjv2,Zhou2026} based on iSWAP gates. The recently introduced directional code~\cite{https://doi.org/10.48550/arxiv.2507.19430} can be viewed as a special case of routing codes, which achieves a higher encoding rate than the surface code while maintaining only nearest-neighbor interactions on a two-dimensional torus. This demonstrates the potential of qubit routing techniques in increasing non-locality, which is crucial for constructing high-rate QLDPC codes. Inspired by this, we go a step further by allowing non-local couplings and relaxing the constraint of identical routing paths. This generalization enables us to find instances with encoding rates approaching those of BB codes.

We first establish a general framework for routing codes, which is defined on more general permutation maps. We further point out that for code constructions where the routing paths of X and Z syndrome qubits satisfy specific symmetries, the commutation condition can be naturally satisfied. In practice, we observe that time-reversal symmetry of the routing sequences yields better performing codes than the identical condition approach previously studied in Ref.~\cite{https://doi.org/10.48550/arxiv.2507.19430}. Guided by this observation, we search for routing codes on a standard torus. Importantly, all these constructions satisfy a key hardware constraint where each code requires only a single non-local coupling vector alongside nearest-neighbor couplings—which fundamentally simplifies the connectivity layout. All identified instances use a vector shorter than $(6,3)$, require a connectivity of only 4 or 5 per qubit. 

For weight-7 routing codes, we already achieve encoding rates competitive with BB codes. Under circuit-level noise, our numerical simulations reveal that weight-7 routing codes outperform surface codes of the same circuit-level distance, yielding a threshold of 0.7\% and lower error rate per logical qubit. Routing codes reduce the physical qubit overhead by approximately a factor of 8 compared to rotated surface codes to reach a similar logical error rate. By further increasing the stabilizer weight up to 11, we find explicit instances such as [[200, 24, 14]] and [[200, 16, 17]] whose encoding rates reach up to over 23 times that of the surface code, at the cost of an increase in the logical error rates. Compared to the directional tile codes in Ref.~\cite{gu2026nearestneighbourgatesneedhighrate}, the weight-7 routing codes already achieve slightly higher encoding rates than weight-11 directional tile codes at the same code distance. When the routing code weight is likewise increased to 11, the parameters $kn / d^2$ can reach over 2.4 times that of directional tile codes.

Crucially, the near-term quantum hardware is rapidly approaching the level of maturity required to implement routing codes. On superconducting processors, high-fidelity non-local two-qubit gates have been demonstrated~\cite{xx24-r7q6,PhysRevLett.134.020801,marxer2023long,wang2025demonstration,Zhang2025,li2025fastrobustremotetwoqubit}, and multi-tier architectures with superconducting bump bonds~\cite{rosenberg20173d,field2024modular,kosen2024signal,norris2025performance} and through-silicon vias (TSVs)~\cite{yost2020solid,mallek2021fabrication,hazard2023characterization} have been separately realized. On neutral-atom platforms, large-scale transport of qubit arrays with acousto-optic deflectors (AOD) is now routine~\cite{Bluvstein2022,bluvstein2024logical,Bluvstein2025}. Moreover, high-fidelity iSWAP gates, the native entangling operation in our construction, have been demonstrated across superconducting, trapped-ion and neutral-atom platforms with fidelities comparable to the widely used CZ gate~\cite{PhysRevX.11.021058,Picard2024,ildefonso2025expandingneutralatomgate}. Together, these advances provide the foundation for routing non-local couplers across multiple tiers and for executing high-fidelity iSWAP-based stabilizer measurements.

Routing codes directly exploit these hardware capabilities through their short, parallel non-local connectivity, enabling a quantitatively simpler hardware layout. On a multi-layer superconducting processor, the parallelism substantially reduces wiring crossings, resulting in a layout with significantly fewer routing tiers and a markedly lower hardware complexity score compared to BB codes of comparable distance. On neutral-atom hardware, the $(2,1)$ coupling vector is three times shorter than $(6,3)$, which proportionally reduces the shuttle distance, while its parallel structure removes row and column conflicts during AOD-based movement. These quantitative advantages strongly suggest that routing codes are a promising candidate for near-term hardware implementation. We therefore believe that routing codes represent a significant step toward bridging the persistent gap between the theoretical optimality and the practical realizability of qLDPC codes.

The remainder of this paper is organized as follows. Section~\ref{sec: construction} presents the general construction of routing codes and demonstrate the special role that the identity and time-reversal symmetry of routing paths play in the code construction. Section~\ref{sec: implementation} describes how routing codes are implemented by qubit routing using iSWAP gates. Section~\ref{sec: instances} provides attrctive code examples and analyzes their circuit-level performance. Section~\ref{sec: hardware} quantitatively analyzes the hardware implementation advantages of routing codes, and Section~\ref{sec: discussion} concludes with a discussion and outlook.

\section{construction} 
\label{sec: construction} 
We now present the formal construction of routing codes. The code can be defined on a standard torus. Let the set of lattice sites be \(L = \mathbb{Z}_l \times \mathbb{Z}_m\) with both \(l, m\) even, and its associated group ring is
\begin{equation}
R = \mathbb{F}_2[x,y]/(x^l-1,\;y^m-1),
\end{equation}
where a site \((i,j)\in L\) corresponds to the monomial \(x^i y^j\). We assign labels to qubits at their initial positions to specify their roles in the error correction. The lattice is partitioned into data qubits \(D = \{(i,j) \mid i+j \text{ is even}\}\), X-syndrome qubits \(X = \{(i,j) \mid i \text{ is odd}, j \text{ is even}\}\), and Z-syndrome qubits \(Z = \{(i,j) \mid i \text{ is even}, j \text{ is odd}\}\).

Next, we define the qubit routing over $T$ time steps. Let \(\{\mathbf{v}_t\}_{t=1}^{T}\) and \(\{\mathbf{w}_t\}_{t=1}^{T}\) be sequences of routing vectors in $L$ for the X- and Z-syndrome qubits, respectively. At step $t$, every X-syndrome qubit at position $\mathbf{r}$ is swapped with the data qubit at $\mathbf{r}+\mathbf{v}_t$, and every Z-syndrome qubit at $\mathbf{r}$ is swapped with the data qubit at $\mathbf{r}+\mathbf{w}_t$. We implicitly assume that the routing vectors are chosen such that the target of an X (or Z) qubit is always a D qubit, and the targets of X and Z routings never conflict.

The routing process at step \(t\) can be concisely described by a permutation map \(\sigma_t\) defined as follows:
\begin{equation}
\sigma_t(\mathbf{u}) =
\begin{cases}
\mathbf{u} + \mathbf{v}_t, & \mathbf{u} \in X,\\
\mathbf{u} + \mathbf{w}_t, & \mathbf{u} \in Z,\\
\mathbf{u} - \mathbf{v}_t, & \mathbf{u} \in D \text{ and } \mathbf{u} - \mathbf{v}_t \in X,\\
\mathbf{u} - \mathbf{w}_t, & \mathbf{u} \in D \text{ and } \mathbf{u} - \mathbf{w}_t \in Z.
\end{cases}
\end{equation}
The four cases partition $L$, so $\sigma_t$ is well defined. Moreover, $\sigma_t^{-1} = \sigma_t$, i.e., each step is an involution.

For each X- or Z-syndrome qubit, the set of data qubits that are successively routed to it over the \(T\) time steps defines an X- or Z-type stabilizer operator. We label these data qubits by their initial positions on the torus. Let an X-syndrome qubit start at $\mathbf{x}$. Its position evolves as $\mathbf{x}^{(t)} = \sigma_t(\mathbf{x}^{(t-1)})$ with $\mathbf{x}^{(0)} = \mathbf{x}$. The initial position of the $t$-th data qubit routed to it is
\begin{equation}
\begin{aligned}
\mathbf{q}_{\mathbf{x},t} 
&= \sigma_1^{-1} \circ \sigma_2^{-1} \circ \cdots \circ \sigma_{t-1}^{-1}(\mathbf{x}^{(t)})  \\
&= (\sigma_1 \circ \sigma_2 \circ \cdots \circ \sigma_{t-1} \circ \sigma_t \circ \sigma_{t-1} \circ \cdots \circ \sigma_1)(\mathbf{x}) \\
&= S_{t-1}^{-1}\, \sigma_t\, S_{t-1}(\mathbf{x}),
\end{aligned}
\end{equation}
where $S_t = \sigma_1 \circ \sigma_2 \circ \cdots \circ \sigma_t$ and we used $\sigma_i^{-1} = \sigma_i$. The stabilizer associated with this X-syndrome qubit is then fully specified by the set of initial data-qubit positions
\begin{equation}
Q_{\mathbf{x}} = \{\, \mathbf{q}_{\mathbf{x},1},\; \mathbf{q}_{\mathbf{x},2},\; \dots,\; \mathbf{q}_{\mathbf{x},T} \,\} \subset L,
\end{equation}
and analogously for a Z-syndrome qubit with initial position \(\mathbf{z}\), we obtain a set \(Q_{\mathbf{z}} = \{\mathbf{q}_{\mathbf{z},1},\dots,\mathbf{q}_{\mathbf{z},T}\}\).

The initial label configuration is invariant under even--even translations. The code is therefore determined by one representative X stabilizer and one representative Z stabilizer. Fixing reference positions $\mathbf{x}$ and $\mathbf{z}$, the entire code is described by the two sets $Q_{\mathbf{x}}$ and $Q_{\mathbf{z}}$. Mapping each site $(i,j)$ to the monomial $x^i y^j$ gives the polynomial representation
\begin{equation}
P(x,y) = \sum_{(a,b) \in Q_{\mathbf{x}}} x^{a} y^{b}, \quad
Q(x,y) = \sum_{(c,d) \in Q_{\mathbf{z}}} x^{c} y^{d},
\end{equation}
in the group ring $R$. All other X and Z stabilizers are then generated by even–even translations of \(P(x,y)\) and \(Q(x,y)\), respectively. The following proposition gives the commutativity conditions for the stabilizers defined as such.

\begin{proposition}
\label{prop:commutativity}
The stabilizer operators defined by above polynomials satisfy commutativity condition if and only if, in the expansion of \(P(x,y)\,\cdot Q(x^{-1},y^{-1})\), all terms \(x^i y^j\) with both \(i\) and \(j\) are even exponent have an even coefficient. 
\end{proposition}

Given the routing vector sequences \(\{v_t\}\) and \(\{w_t\}\), verifying the commutativity condition in Proposition~\ref{prop:commutativity} generally requires first computing \( P(x,y)\) and \(Q(x,y)\). However, the following proposition shows that commutativity is automatically guaranteed when the two sequences are same and time-reversal symmetric. This yields a direct construction of routing codes.

\begin{proposition}
\label{prop:tr-symmetric}
Let \(\{v_t\}_{t=1}^{T}\) be a sequence of routing vectors satisfying
\begin{enumerate}
    \item \(v_t = w_t\) for all \(t = 1,\dots,T\), i.e., X and Z syndrome qubits follow identical routing paths;
    \item \(v_t = v_ {\bar{t}} \coloneqq v_{T-t+1}\) for all \(t\), i.e., the routing sequence is time-reversal symmetric.
\end{enumerate}
Then the stabilizer operators of the resulting routing code automatically satisfy commutativity condition.
\end{proposition}

The two symmetry conditions in Proposition~\ref{prop:tr-symmetric} should be viewed as special construction principles rather than necessary conditions for obtaining a valid routing code. When the X- and Z-syndrome qubits follow identical routing paths and the routing sequence is time-reversal symmetric, the commutation relations follow automatically from Proposition~\ref{prop:tr-symmetric}. Outside this symmetric subclass, the commutativity and temporal-ordering conditions must be verified explicitly for each instance. As will be seen shortly in Section~\ref{sec: instances}, in our search, we independently explored constructions with identical paths and constructions with time-reversal symmetry, as well as instances outside the simultaneous-symmetry subclass. Interestingly, the best-performing instances identified in our search consistently arose from the latter class, suggesting that relaxing the time-reversal symmetry can enlarge the design space and lead to improved code performance. We also note that the routing sequences employed in Refs.~\cite{https://doi.org/10.48550/arxiv.2507.19430,gu2026nearestneighbourgatesneedhighrate} also satisfy both conditions of Proposition~\ref{prop:tr-symmetric}, even though their search was conducted solely based on the identity condition. This convergence suggests that the design principles underlying routing codes are not only natural but also broadly applicable, providing a unified perspective that encompasses both direction codes and the present constructions.

\section{Implementation by qubit routing}
\label{sec: implementation}

The algebraic construction of routing codes admits a direct hardware implementation using qubit routing~\cite{c31l-qjv2,Zhou2026}. The fundamental primitive is the iSWAP gate, which, up to single-qubit gates, is equivalent to a CNOT followed by a SWAP (a CXSWAP operation). Replacing every CNOT in a standard syndrome-extraction circuit with an iSWAP preserves the logical action while swapping the syndrome qubit with the data qubit it interacts with. The permutation $\sigma_t$ is thus directly realized as a physical circuit, and the connectivity requirement is reduced from the full Tanner graph to only the displacement vectors $\{\mathbf{v}_t, \mathbf{w}_t\}$ that appear in the routing sequence.

Concretely, for an X-syndrome qubit initially at position $\mathbf{x}$, the step corresponding to routing vector $\mathbf{v}_t$ is an iSWAP gate between the current syndrome qubit and the data qubit at $\mathbf{x}^{(t-1)}+\mathbf{v}_t$. After the gate, the syndrome qubit moves to $\mathbf{x}^{(t)} = \mathbf{x}^{(t-1)}+\mathbf{v}_t$, and the data qubit moves to the syndrome location. By executing these gates in sequence for $t=1,\dots,T$, the syndrome qubit traverses the path dictated by the routing vectors, collecting parity information from each visited data qubit. The same procedure applies to Z-syndrome qubits using the vectors $\mathbf{w}_t$.

Since the routing permutations $\sigma_t$ are defined such that $X$ and $Z$ syndrome qubits never target the same data qubit at the same time step , all CXSWAP gates within a single time step can be applied in parallel across all syndrome qubits without conflict. Consequently, the full syndrome-extraction cycle consists of exactly $T$ parallel layers of CXSWAP gates, matching the number of routing vectors. 

For the fully parallelized implementation to be valid, an additional temporal condition must be satisfied. Consider the shared data qubits of an X and a Z stabilizer. Among these, the subset for which the X-syndrome qubit acts before the Z-syndrome qubit ($t<s$) must also have even cardinality~\cite{PRXQuantum.5.010348}. In terms of Eq.~\eqref{eq:delta_diff_condition}, for every even--even vector $(2u,2v)$, the number of solutions with $t<s$ is even. As an important special case, the following proposition shows that the temporal condition is automatically guaranteed when the routing vectors are non-negative and the torus is chosen sufficiently large.

\begin{proposition}
\label{prop:parallel_commute}
Consider the time-reversal symmetric construction of Proposition~\ref{prop:tr-symmetric}, with identical X and Z routing vectors $v_t = w_t$ ($t=1,\dots,T$) whose components are non-negative. On a torus $\mathbb{Z}_l \times \mathbb{Z}_m$ sufficiently large such that no path wraps around, i.e.,
\begin{equation}
\sum_{t=1}^{T} v_t < (l, m) \quad \text{(component-wise)},
\end{equation}
the temporal condition for fully parallelized implementation is automatically satisfied.
\end{proposition}

After a complete forward routing cycle, the positions of syndrome and data qubits are generally permuted relative to their initial configuration. To repeat the measurement of the same stabilizers, a backward routing cycle is applied immediately afterwards, using the reversed sequence of vectors $\{-\mathbf{v}_{T}, \dots, -\mathbf{v}_1\}$ and $\{-\mathbf{w}_{T}, \dots, -\mathbf{w}_1\}$, which returns all qubits to their original positions. By alternating forward and backward cycles, one can perform repetitive syndrome extraction, following the standard scheme for dynamical quantum error-correcting codes~\cite{McEwen2023relaxinghardware}.

\section{Explicit instances and circuit-level performance}
\label{sec: instances}

We focus first on the weight-7 subfamily, as it shares the same circuit depth as the weight-6 BB code and thus may achieve a comparable threshold. Table~\ref{tab:selected_codes} presents explicit weight-7 routing code examples with encoding rates competitive with BB codes and good circuit-level performance. To obtain these codes presented, we have searched for codes following the construction of Proposition~\ref{prop:tr-symmetric}, where the two sequences are identical and time-reversal symmetric. We also searched for codes whose sequences satisfy only one of these two conditions, and separately verified the commutativity and circuit feasibility. The examples listed further restrict the sequences to a single non-local coupling vector shorter than \((6,3)\), and the local vectors \((0,1)\) and \((1,0)\). All instances are defined on a standard torus, use non-local vectors shorter than $(6,3)$, and require a connectivity of only 4 or 5 (see Fig.~\ref{fig:combined}a for an example). More importantly, the all-parallel geometry of the non-local couplings eliminates wiring crossings at the structural level.

\begin{table*}[htbp]
\centering
\renewcommand{\arraystretch}{1.15}
\small
\begin{tabular}{c c c c c c c}
\toprule
$\,$\textbf{Code} $\,$&$\,$ \textbf{$[[n,k,d]]$} $\,$&$\,$ $d_{\text{circ}}$ $\,$&$\,$ \textbf{Encoding rate} $\,$&$\,$ \textbf{Connectivity (local + non-local)} $\,$&$\,$ \textbf{Non-local coupling} $\,$&$\,$ \textbf{Circuit depth} $\,$\\
\midrule
\multirow{8}{*}{\shortstack{Routing codes}} 
& $[[54,8,6]]$  & $\leq 6$ & $1 / 6.75$ & ${3+1}$ & ${(3,0)}$  & 7\\
& $[[70,8,7]]$  & $\leq 7$ & $1 / 8.75$ & ${3+1}$ & ${(2,1)}$  & 7\\
& $[[80,8,7]]$  & $\leq 8$ & $1 / 10$ & ${3+2}$ & ${(4,1)}$  & 7\\
& $[[90,8,9]]$  & $\leq 9$ & $1 / 10$ & ${3+2}$ & ${(4,1)}$  & 7\\
& $[[100,8,10]]$ & $\leq 9$ & ${1 / 12.5}$ & ${3+2}$ & ${(2,1)}$  & 7\\
& $[[110,8,11]]$ & $\leq 10$ & $1 / 13.75$& ${3+2}$ & ${(6,1)}$  & 7\\
& $[[140,8,12]]$ & $\leq 11$ & ${1 / 17.5}$ & ${3+1}$ & ${(4,1)}$  & 7\\
& $[[140,8,13]]$ & $\leq 12$ & ${1 / 17.5}$ & ${3+2}$ & ${(6,1)}$  & 7\\

\midrule
\multirow{3}{*}{\shortstack{BB codes~\cite{bravyi2024high}}}
& $[[72,12,6]]$  & $\leq 6$ & $1 / 6$    & $4+2$ & $(6, 3)$, $(3, 6)$  & 7\\
& $[[108,8,10]]$ & $\leq 8$ & $1 / 13.5$ & $4+2$ & $(6, 3)$, $(3, 6)$  & 7\\
& $[[144,12,12]]$ & $\leq 10$ & $1 / 12$   & $4+2$ & $(6, 3)$, $(3, 6)$  & 7\\
\bottomrule
\end{tabular}
\caption{Weight-7 routing code instances with distances $6$ to $13$, compared to BB codes that share the same generating polynomials as the gross code. For each code we list the parameters $[[n,k,d]]$, the estimated upper bound on the circuit-level distance $d_{\text{circ}}$, the encoding rate expressed as $1/x$ (the number of physical qubits per logical qubit), the qubit connectivity (local plus non-local couplers per qubit), the non-local coupling vector, and the two-qubit gate depth of one error-correction cycle. The upper bounds on $d_{\text{circ}}$ are obtained by randomized sampling with over $10^6$ trials. All routing codes use a single non-local vector with a length shorter than $(6,3)$ and require connectivity of only 4 or 5, compared to 6 for the BB codes. The two-qubit gate depth of the error correction cycle of routing codes equals the stabilizer weight of 7.}
\label{tab:selected_codes}
\end{table*}

\begin{figure*}[htpb]
    \centering
    \includegraphics[width=\textwidth]{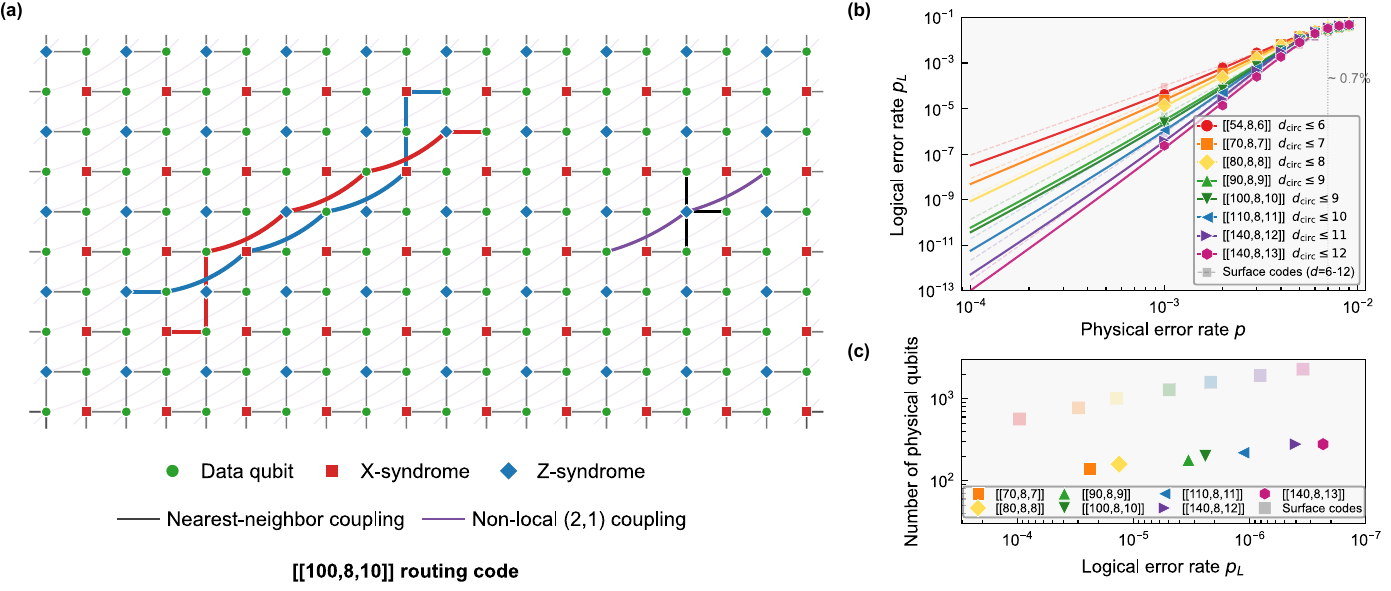}
    \caption{%
        \textbf{(a)} Hardware connectivity of the $[[100,8,10]]$ routing code on a $20\times10$ torus.
        Green circles, red squares, and blue diamonds represent the initial position of data, $X$-, and $Z$-syndrome qubits, respectively. Gray lines show nearest-neighbor couplings and faint extensions at the edges indicate the periodic boundary conditions.
        Purple curves represent the internal $(2,1)$ non-local couplings, displayed with high transparency for visual clarity. The figure also shows the routing paths of the X and Z syndrome qubits, which exhibit time-reversal symmetry. Additionally, the connection pattern of a representative qubit is highlighted on the right.
        \textbf{(b)} Average Logical error rate per round per logical qubit as a function of the physical error rate $p$. Solid lines and colored markers represent routing codes under circuit-level depolarizing noise and dashed lines and square markers represent rotated surface codes ($d$ from 6 to 13) under the same noise model. For surface codes, the color matches the routing code of the corresponding circuit-level distance, but with higher transparency for visual distinction. Curves are fits to $p_L = p^{d_{\mathrm{circ}}/2} \exp(c_0 + c_1 p + c_2 p^2)$, where $d_{\mathrm{circ}}$ is the estimated upper bound on the circuit-level distance.
        \textbf{(c)} Physical qubit overhead versus logical error rate at physical error $p = 10^{-3}$.
        Colored symbols represent routing codes and squares with lower opacity show the surface codes at the matching distance. The number of physical qubits for surface codes is scaled to $8 \times (2d^2-1)$ to encode eight logical qubits, matching the routing codes.
    }
    \label{fig:combined}
\end{figure*}

We estimate an upper bound on the circuit-level distance $d_{\text{circ}}$ by the randomized sampling method of Ref.~\cite{bravyi2024high}. The estimated upper bounds for routing codes are higher than those for BB codes of the same theoretical distance, suggesting that routing codes may scale more favorably to lower logical error rates. While this does not constitute a rigorous proof of a larger circuit-level distance, no exceptions to this trend were observed across over $10^6$ trials.

 We then evaluate the performance of weight-7 routing codes under circuit-level noise. We simulate these codes in Table~\ref{tab:selected_codes} with physical error rates $p$ from $10^{-2}$ to $10^{-3}$, using the circuit-level depolarizing noise model (see Appendix~\ref{simulation}). Under the same estimated circuit-level distance, routing codes achieve a logical error rate lower than that of the rotated surface code. As shown in Fig.~\ref{fig:combined}b, the routing codes exhibit a threshold of approximately 0.7\%, comparable to that of surface codes and BB codes. By extrapolating the logical error rate curve to $p = 10^{-4}$, it is observed that the logical error rate of the distance-13 routing code is already lower than $10^{-13}$, using only 280 physical qubits to encode 8 logical qubits. For all codes compared here, we use the same BP-OSD decoder~\cite{Panteleev2021} to ensure a fair comparison at the decoder level. This choice lowers the performance of surface codes relative to values commonly reported with the minimum-weight perfect matching decoder~\cite{PhysRevA.83.020302,Higgott2022}, but it ensures a fair comparison, since the performance of routing codes can also be further improved with more accurate decoders, such as relay-BP~\cite{müller2025improvedbeliefpropagationsufficient} or neural-network-based decoders~\cite{Bausch2024,senior2026scalablerealtimeneuraldecoder}.

Based on circuit-level simulations, We estimate the physical qubit overhead reduction relative to the rotated surface code. Using surface code data at distances $6$ to $12$ and $p = 10^{-3}$, we identify, for each routing code, the surface code distance whose logical error rate most closely matches that of the routing code. As shown in Fig.~\ref{fig:combined}c, routing codes reduce the physical qubit overhead by around a factor of 8 compared to the surface code to reach a close logical error rate.

Further, by extending the length of routing sequences, one can search for instances whose encoding rates significantly surpass those of BB codes. Fig.~\ref{fig2}a displays code instances with weights ranging from 8 to 11, along with their encoding rate multiples $kd^2/n$ relative to the surface code. Note that we have specifically checked the upper bound on the circuit-level distance for these codes, which is at most one less than the theoretical distance, implying that the actual circuit-level performance will not deteriorate sharply. When the weight is increased to 11, the encoding rate of routing codes [[200, 24, 14]] and [[200, 16, 17]] can reach over 23 times that of the surface codes with the same distances. Compared to the purely nearest-neighbor directional tile codes in Ref.~\cite{gu2026nearestneighbourgatesneedhighrate}, weight-7 routing codes already achieve encoding rates comparable to their weight-11 codes at the same distance. When the weight is increased to 11, the $kd^2/n$ parameter of the [[200, 24, 14]] code reaches over 2.4 times that of weight-11 directional tile codes. However, since the stabilizer weight directly determines the depth of the error-correction circuit and is thereby closely tied to the logical error rates, the gain in encoding rate comes at the cost of increasing logical error rates. The logical error rate tests for routing codes with distances from 6 to 13 at a physical error rate $p=0.001$ in Fig.~\ref{fig2}b corroborate this analysis. As the weight increases, the ratio of the logical error rate of these codes to that of the surface code with the same upper bound on the circuit-level distance steadily rises. This reflects a trade-off among the encoding rate, the stabilizer weight, and the logical error rate. 

\begin{figure}[htpb]
    \centering
    \includegraphics[width=8.5cm]{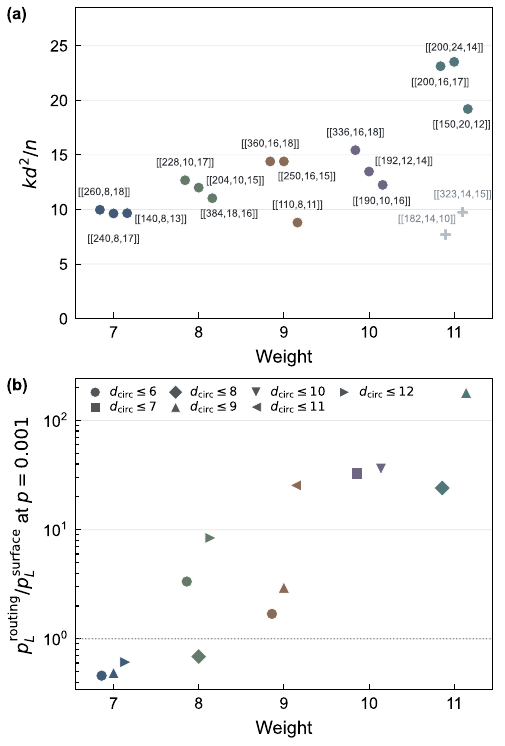}
\caption{{Encoding efficiency and logical performance in routing codes with different weights.} \textbf{(a)} Encoding-efficiency metric \(kd^{2}/n\) for representative routing codes with stabilizer weights \(w=7\)--\(11\). For each weight, three routing codes are shown and labeled by their parameters \([[n,k,d]]\). The light-colored crosses denote the weight-11 directional tile codes reported in Ref.~\cite{gu2026nearestneighbourgatesneedhighrate}. \textbf{(b)} Ratio of the logical error rate of each routing code to that of a surface code with the same upper bound on the circuit-level distance, evaluated at a physical error rate \(p=10^{-3}\). Colors indicate the stabilizer weight, while marker shapes distinguish the circuit-level distance \(d_{\mathrm{circ}}\). The horizontal dotted line marks equal logical error rates, \(p_{L}^{\mathrm{routing}}/p_{L}^{\mathrm{surface}}=1\). Complete information on all codes shown in the figure can be found in Appendix~\ref{sec:more_instances}.}
    \label{fig2}
\end{figure}

\section{Hardware implementation advantages}
\label{sec: hardware}

The structural properties of routing codes translate into concrete and, in several respects, decisive advantages for physical implementation on superconducting and neutral-atom platforms. 

In superconducting platforms, the reduced qubit connectivity (4 or 5) of routing codes directly alleviates frequency crowding and the associated crosstalk. This lower connectivity permits higher couplings, reduces on-chip structures and control lines, and simplifies frequency assignment to lower the risk of parasitic resonances ~\cite{PhysRevApplied.10.054062, kosen2024signal}.

Beyond connectivity, the parallel geometry of the non-local couplings offers a deeper architectural advantage. To realize the non-local couplings required by qLDPC codes, superconducting platforms are increasingly considering multi-layer architectures. In a multi-chip stackup, additional routing chips can be bonded above the qubit chip via TSVs, with each additional chip providing two opposing routing layers~\cite{yost2020solid,rosenberg20173d,field2024modular}. High-fidelity two-qubit gates mediated by multi-bump couplers have already been demonstrated~\cite{field2024modular,norris2025performance}, and long-range on-chip couplers are under active development~\cite{wang2025demonstration,kumph2024demonstration,marxer2023long,Zhang2025,li2025fastrobustremotetwoqubit}. In this setting, because all non-local edges in a routing code share the same direction, they generate far fewer crossings that must be resolved by detours or by promoting edges to higher tiers. This directly reduces the number of routing tiers required and shortens the average coupler length, which in turn lowers the parasitic capacitance and fabrication complexity associated with additional TSVs and bump bonds.

To quantify these advantages, we evaluate the hardware resource overhead using the Hardware-Aware Layout (HAL) framework~\cite{Mathews2026}, which automates the place-and-route process for multi-layer superconducting architectures. HAL takes a code's connectivity graph as input and outputs the required number of routing tiers, the average coupler length, the average number of bump bonds per edge, the average number of TSVs per edge, and a composite hardware complexity score. We apply HAL to all routing codes from Table~\ref{tab:selected_codes} and to the corresponding BB codes for comparison. The results are summarized in Table~\ref{tab:hal_comparison}. Note that weight-7 routing codes as representative examples is sufficient, since higher-weight routing codes do not introduce additional hardware connectivity. Across all metrics, routing codes require substantially fewer hardware resources. The number of routing tiers is consistently lower (3--4 versus 5--6 for BB codes), and the average coupler length, average bump bonds per edge, and average TSVs per edge are all reduced. The composite hardware complexity score, which aggregates these metrics, remains below 2.01 for all routing codes, reaching as low as 1.57, whereas all BB codes exceed 2. It should be noted that although all non-local couplings in routing codes are parallel, this parallelism holds on the torus topology. The boundary connections that arise when mapping the torus to a plane pose the greatest limitation to realizing the layout within only two layers.

\begin{table}[t]
\centering
\renewcommand{\arraystretch}{1.15}
\small
\begin{tabular}{c c c c c c}
\toprule
\,$[[n,k,d]]$ \,&\, Tiers \,&\, Length \,& \, Bumps \,& TSVs & \,$C_\text{hw}$ \,\\
\midrule
\multicolumn{6}{c}{\textbf{Routing codes} } \\
$[[54,8,6]]$   & 3 & 4.63  & 2.73 & 2.06 & 1.57 \\
$[[70,8,7]]$   & 3 & 6.22  & 3.14 & 2.07 & 1.64 \\
$[[80,8,8]]$   & 4 & 8.62  & 3.38 & 2.34 & 1.81 \\
$[[90,8,9]]$   & 3 & 7.43  & 3.41 & 2.25 & 1.70 \\
$[[100,8,10]]$ & 4 & 11.43 & 2.93 & 2.25 & 1.85 \\
$[[110,8,11]]$ & 4 & 10.63 & 4.71 & 2.58 & 1.96 \\
$[[140,8,12]]$ & 3 & 7.43  & 3.07 & 2.19 & 1.68 \\
$[[140,8,13]]$ & 5 & 10.32  & 4.32 & 2.72 & 2.01 \\
\midrule
\multicolumn{6}{c}{\textbf{BB codes}~\cite{bravyi2024high}} \\
$[[72,12,6]]$   & 5 & 8.94  & 4.67 & 2.94 & 2.01 \\
$[[108,8,10]]$  & 5 & 13.54 & 4.20 & 2.73 & 2.09 \\
$[[144,12,12]]$ & 6 & 14.58 & 4.65 & 3.17 & 2.24 \\
\bottomrule
\end{tabular}
\caption{Hardware resource evaluation using HAL~\cite{Mathews2026}. For each code we report the number of tiers, the average coupler length, the average number of bump bonds per edge, the average number of TSVs per edge, and the composite hardware complexity score $C_\text{hw}$. All results are obtained under the default parameter settings of HAL, taking the best case among the available layout options. All routing codes achieve lower complexity scores than the BB codes of comparable distance.}
\label{tab:hal_comparison}
\end{table}

In neutral-atom platforms, We consider the zoned neutral-atom architecture, where qubits are stored in a static optical lattice and selectively transported to an entangling region by moving optical tweezers controlled by orthogonal AODs~\cite{bluvstein2024logical,Bluvstein2025}. The shuttling process introduces two primary error sources, motional heating and atom loss, both of which grow with transport time and distance.

The BB codes require coupling vectors such as $(6,3)$, with a geometric distance of $d_{(6,3)}\approx 6.71$. Routing codes restrict all non-local couplings to vectors as short as $(2,1)$, yielding $d_{(2,1)}\approx 2.24$. This $\sim\!3\times$ reduction in shuttle distance translates directly to a proportionally shorter transport time. Since both motional heating and atom loss depend on the transport duration~\cite{Bluvstein2022}, the shorter distance significantly reduces the contribution of these error sources.

The parallel geometry of the long-range edges offers an additional scheduling advantage. Since AOD rows and columns cannot cross during transport~\cite{Stade2024}, multi-directional vectors in BB codes cause conflicts that force serialization of long-range gates. Because of the torus topology, routing codes require two time steps to execute one layer of long-range gates without AOD conflicts, while BB codes require up to four and typically two. The combination of shorter shuttle distances and simplified scheduling collectively reduces the QEC cycle time and operational error rate on neutral-atom hardware.

\section{Discussion}
\label{sec: discussion}
We have introduced routing codes, a new family of qLDPC codes constructed by iSWAP-based qubit routing, which simultaneously deliver high rates, high threshold, low qubit connectivity, short and parallel non-local coupling. Combined with encoding rates competitive with BB codes and reduced qubit connectivity, routing codes achieve a level of hardware simplicity that existing qLDPC constructions, even with circuit-level optimizations, cannot reach.

We emphasize that the explicit codes reported here represent only a small fraction of the full routing code landscape. Several promising directions remain open. First, the construction can be extended beyond the standard torus topology to twisted tori, where twisted boundary conditions may further improve the encoding rate~\cite{rmy6-9n89,https://doi.org/10.48550/arxiv.2507.19430}. Second, relaxing the restriction to a single non-local coupling vector and allowing multiple non-local vectors within a sequence could yield new structures with even better parameters.

As we have already pointed out, constructing planar routing codes is a key step toward further reducing hardware complexity, as it would enable completely conflict-free and shortest wiring within a two-layer architecture. We anticipate that routing codes can be planarized using methods similar to those in Ref.~\cite{gu2026nearestneighbourgatesneedhighrate} within the framework of tile codes. Such routing codes, with their modestly enriched connectivity structure, may offer higher encoding rates than directional tile codes.

More generally, the routing code concept is not tied to translation-invariant codes. It can be defined on any bipartite graph, with each syndrome qubit following its own individualized routing path. This also opens the possibility of hardware-adaptive quantum error correction. On a chip with defective qubits, where the connectivity graph loses its translational invariance, geometrically regular codes such as the surface code or BB codes become difficult to implement directly~\cite{PhysRevApplied.19.064081,Siegel2023adaptivesurfacecode,wei2026adaptivedeformationcolorcode}. Routing codes, in contrast, can be tailored to the irregular connectivity by designing custom routing paths that avoid the defective regions, thereby enabling fault-tolerant encoding on imperfect hardware without additional defect-mitigation overhead.

Meanwhile, the realization of fault-tolerant logical operations on routing codes (and on dynamical codes in general) remains a open problem. The predominant technique of code surgery has so far been developed mainly for static codes~\cite{Cohen2022}. We anticipate that incorporating analogous dynamical code structures into the ancilla regions during surgery will be essential for maintaining consistency with the rest of the code.

\begin{acknowledgments}
This work was supported by the National Key Research and Development Program of China (Grant No. 2023YFB4502500), Anhui Province Innovation Plan for Science and Technology (Grant No. 202423r06050002), and the China Postdoctoral Science Foundation (Grant No. 2026M793680).
\end{acknowledgments}

\begin{appendix}

\section{Proof of propositions}
\textbf{Proof of Proposition 1}. Two stabilizers commute if and only if \(|Q_{\mathbf{x}} \cap Q_{\mathbf{z}}|\) is even. The intersection size \(|Q_{\mathbf{x}} \cap Q_{\mathbf{z}}|\) is simply the number of pairs \((t, s)\) satisfying the vector equation
\begin{equation}
(a_t, b_t) = (c_s, d_s) \pmod{(l,m)}.
\end{equation}
where $(a_t, b_t) \in Q_{\mathbf{x}}$ and $(c_t, d_t) \in Q_{\mathbf{z}}$. Commutativity requires the total number of such solutions are even.

Now, by translation invariance, the code contains all even–even translates of the Z stabilizer. Thus we must require that for every even–even vector \((2u, 2v)\), the shifted Z stabilizer \(Q_{\mathbf{z}} + (2u, 2v)\) also has even overlap with \(Q_{\mathbf{x}}\). The number of solutions to
\begin{equation}\label{eq:delta_diff_condition}
(a_t-c_s, b_t-d_s) = (2u, 2v) \pmod{(l,m)}
\end{equation}
must therefore be even for all \((2u, 2v)\).

This is exactly the information encoded in the group-ring product. Consider the product \(P(x,y) Q(x^{-1}, y^{-1})\). Its expansion is
\begin{equation}
P(x,y) Q(x^{-1}, y^{-1}) = \sum_{t,s} x^{a_t - c_s} y^{b_t - d_s}.
\end{equation}
For a fixed even–even vector \((2u, 2v)\), the coefficient of the term \(x^{2u} y^{2v}\) in this sum counts, modulo 2, exactly the number of pairs \((t,s)\) such that \((a_t - c_s, b_t - d_s) = (2u, 2v)\). Therefore, commutativity with all Z stabilizers holds if and only if, in the expansion of \(P(x,y) Q(x^{-1}, y^{-1})\), every term \(x^i y^j\) with both \(i\) and \(j\) even has an even coefficient. $\square$

\textbf{Proof of Proposition 2}. Under the condition $v_t = w_t$, the X and Z syndrome qubits follow analogous routing paths satifying $Q_{\mathbf{x}}-\mathbf{x} = Q_{\mathbf{z}}  - \mathbf{z}$, where $\mathbf{x}$ and $\mathbf{z}$ are initial positions of the X and Z syndrome qubits. By the definition of the routing permutation, one can explicitly compute these vectors. A straightforward induction shows that the $t$-th data qubit initially at $\mathbf{x}$ is routed to the syndrome qubit after $t$ steps, and its initial position is~\cite{https://doi.org/10.48550/arxiv.2507.19430}
\begin{equation}
\delta_t \coloneqq \mathbf{q}_{\mathbf{x},t} - \mathbf{x} = 2\sum_{\tau=1}^{t-1} v_\tau + v_t \pmod{(l,m)}.
\end{equation}
The same expression holds for $\delta_t =\mathbf{q}_{\mathbf{z},t} - \mathbf{z}$.

By Proposition~\ref{prop:commutativity}, the stabilizers commute if and only if for every even--even vector $(2u,2v)$ the number of ordered pairs $(t_1,t_2)$ such that
\begin{equation}\label{eq:delta_diff}
(\delta_{t_1} + \mathbf{x}) - (\delta_{t_2}+\mathbf{z}) = (2u,2v) \pmod{(l,m)}
\end{equation}
is even.

Now impose the time-reversal symmetry $v_t = v_{T-t+1}$. We claim that for any solution $(t_1,t_2)$ to \eqref{eq:delta_diff}, the reversed pair $(\bar{t}_1,\bar{t}_2) = (T-t_1+1,\, T-t_2+1)$ is also a solution. Note that these two sets of solutions must be different, since the parity condition forbids \(t_1 = t_2 = (T+1)/2\). Therefore the solutions pair up and their total number is even.

To verify the claim, substitute the symmetry condition $v_\tau = v_{T-\tau+1}$ into the definition of $\delta_t$. A direct computation yields the relation
\begin{equation}
\delta_t + \delta_{T-t+1} = 2S,
\end{equation}
where $S = \sum_{\tau=1}^{T} v_\tau$. Consequently, for any $t_1, t_2$,
\begin{equation}
\delta_{T-t_1+1} - \delta_{T-t_2+1} = -(\delta_{t_1} - \delta_{t_2}),
\end{equation}
from which it follows immediately that $(T-t_1+1, T-t_2+1)$ satisfies equation~\eqref{eq:delta_diff} whenever $(t_1, t_2)$ does. Thus, under the given symmetry, the commutativity condition is automatically fulfilled. $\square$

\textbf{Proof of Proposition 3}. Since the X and Z syndrome qubits follow identical routing sequences, the ordered sets of data qubits they collect,
\(Q_X = \{\mathbf{q}_{\mathbf{x},1},\dots,\mathbf{q}_{\mathbf{x},T}\}\) and
\(Q_Z = \{\mathbf{q}_{\mathbf{z},1},\dots,\mathbf{q}_{\mathbf{z},T}\}\),
are related by an overall odd--odd translation \(\Delta = \mathbf{x} - \mathbf{z}\). That is, \(\mathbf{q}_{\mathbf{z},t} = \mathbf{q}_{\mathbf{x},t} + \Delta\) for all \(t\).

The commutativity condition concerns the intersection of \(Q_X\) and a translate of \(Q_Z\) by an even--even vector \((2u,2v)\). Each element in the intersection corresponds to a pair \((t,s)\) such that
\begin{equation}~\label{eq:delta_diff_condition2}
\mathbf{q}_{\mathbf{x},t} - \mathbf{q}_{\mathbf{z},s} = (2u,2v),
\end{equation}
which is equivalent to
\begin{equation}
\mathbf{q}_{\mathbf{x},t} - \mathbf{q}_{\mathbf{x},s} = (2u,2v) + \Delta.
\end{equation}
We must show that among all such pairs, those with \(t<s\) are even in number.

Because the components of each \(v_t\) are non-negative and the torus is large enough to prevent wrapping, the mapping \(t \mapsto \mathbf{q}_{\mathbf{x}_0,t}\) is strictly increasing in both coordinates. Consequently, for any two distinct indices \(t \neq s\),
\begin{equation}
\mathbf{q}_{\mathbf{x},t} - \mathbf{q}_{\mathbf{x},s} \text{ has  positive coordinate} \iff t > s.
\end{equation}
Now let \(\mathbf{d} = (2u,2v) + \Delta\) be the fixed vector in Eq.~\eqref{eq:delta_diff_condition2}. Note that \(t\neq s\), since \((2u,2v) + \Delta \neq (0,0)\). Therefore, for every solution \((t,s)\) we have either \(t>s\) (if some coordinate of \(\mathbf{d}\) are positive) or \(t<s\) (if some coordinate of \(\mathbf{d}\) are negative). By Proposition~\ref{prop:commutativity}, the total number of solutions is even. Restricting to \(t<s\) simply selects all solutions when \(\mathbf{d}\) has negative coordinates, and none when \(\mathbf{d}\) has positive coordinates. In either case the count is even. $\square$

\section{Qubit routing with iSWAP gates}

In prior works, qubit routing is utilized to efficiently overcome hardware connectivity constraints~\cite{c31l-qjv2,Zhou2026}. The central building block of qubit routing is the iSWAP gate, which is equivalent to a CNOT gate followed by a SWAP gate. This operation commonly referred to as CXSWAP.  Formally, the two-qubit gate
\begin{equation}
\begin{aligned}
\text{CXSWAP} & = \text{SWAP} \times \text{CNOT} \\ 
& =(H \otimes I) (S^\dagger \otimes S^\dagger) \, \text{iSWAP} \, (I \otimes H),
\end{aligned}
\end{equation}
shares the equivalent KAK decomposition~\cite{khaneja2000cartandecompositionsu2nconstructive,tucci2005introductioncartanskakdecomposition} as the iSWAP gate. Consequently, any quantum circuit written with CXSWAP gates can be locally transformed into a circuit of iSWAP gates and vice versa, without changing the overall two-qubit gate count.  In our syndrome-extraction circuits we exploit this equivalence. Each physical iSWAP simultaneously performs a parity check between the data qubit and the syndrome qubit and swaps their positions. The sequential application of such gates moves the syndrome qubit along the prescribed routing path, collecting parity information from every visited data qubit.

\section{Estimating code distances}

\subsection{Theoretical code distances}

Obviously, the routing code proposed belongs to the quantum CSS codes~\cite{nielsen2010quantum}. The number of logical qubits encoded by a CSS code is given by
\begin{equation}
k = n - \operatorname{rank}(H_X) - \operatorname{rank}(H_Z),
\end{equation}
where $H_X$ and $H_Z$ are the parity-check matrices of the $X$ and $Z$ stabilizers, respectively, and $n$ is the total number of physical qubits.  Computing the distance $d$ of a quantum LDPC code is, in general, NP-hard.  Here we formulate the distance calculation as an integer linear programming problem.

Let $H_Z$ be the $n_Z \times n$ parity-check matrix of the $Z$ stabilizers.  For a given logical $Z$ operator, represented as an $n$-dimensional binary vector $\bar{Z}$, we seek the minimum-weight $X$ operator, encoded as a binary vector $x \in \{0,1\}^n$, that commutes with all $Z$ stabilizers while anti-commuting with $\bar{Z}$.  The commutation condition is
\begin{equation}
H_Z \, x = 0 \pmod{2},
\end{equation}
which is equivalent to demanding the existence of an auxiliary integer vector $y \in \mathbb{Z}^{n_Z}$ such that
\begin{equation}
H_Z \, x + 2y = 0.
\end{equation}
The anti-commutation condition with the logical operator reads
\begin{equation}
\bar{Z} \cdot x = 1 \pmod{2},
\end{equation}
which is enforced by requiring an integer variable $z$ satisfying
\begin{equation}
\bar{Z} \cdot x + 2z = 1.
\end{equation}

The integer programming problem therefore reads
\begin{equation}
\begin{aligned}
\text{minimize} \quad & \sum_{i=1}^{n} x_i \\
\text{subject to} \quad & H_Z \, x + 2y = 0, \\
& \bar{Z} \cdot x + 2z = 1, \\
& x_i \in \{0,1\}, \quad y \in \mathbb{Z}^{n_Z}, \quad z \in \mathbb{Z}.
\end{aligned}
\end{equation}

The optimal objective value is the weight of the lightest $X$ operator that satisfies the conditions for the chosen $\bar{Z}$, i.e., the $X$-distance contribution of that logical operator.  Repeating the procedure for all $k$ independent logical $Z$ operators and taking the minimum yields the $X$-distance $d_X$ of the code.  By exchanging the roles of $X$ and $Z$ (substituting $H_X$ and a logical $X$ operator $\bar{X}$) we analogously obtain $d_Z$.  The overall code distance is then $d = \min(d_X, d_Z)$.

\subsection{Circuit-Level Distance}
\label{sec:circuit_distance}
In circuit-level simulations, errors can propagate through entangling gates and measurements, making the effective code distance dependent on the specific syndrome-extraction circuit. The circuit-level distance \(d_{\mathrm{circ}}\) is defined as the minimum number of faulty operations in the syndrome measurement circuit that can cause an undetectable logical error. To estimate the circuit-level distance $d_{\text{circ}}$ of the constructed syndrome measurement circuit, we adopt the BP-OSD based upper bound method proposed in Ref.~\cite{bravyi2024high}. The circuit-level distance $d_{\text{circ}}^Z$ for $Z$-type errors is defined as the minimum number of faulty operations in the syndrome measurement circuit that can generate an undetectable $Z$-type logical error; $d_{\text{circ}}^X$ is defined analogously, and $d_{\text{circ}} = \min(d_{\text{circ}}^X, d_{\text{circ}}^Z)$.

We first construct the circuit-level decoding matrix $D$ using a linearized noise model. For each possible single-fault location in the circuit, including CNOT gates (15 two-qubit Pauli errors), qubit initializations and measurements (bit-flip errors), and idle qubits ($X$, $Y$, or $Z$ errors)—we simulate the fault by propagating the resulting Pauli error through the remaining circuit operations using the stabilizer formalism. This yields, for each fault $j$, the measured syndrome $\mathbf{s}_j^U$, the final data-qubit error syndrome $\mathbf{s}_j^F$, and the logical syndrome $\mathbf{s}_j^L$. The decoding matrix $D$ is formed by stacking $[\mathbf{s}_j^U; \mathbf{s}_j^F]$ as columns, the logical matrix $D^L$ by stacking $\mathbf{s}_j^L$ as columns, and columns with identical syndrome triples are merged with probabilities summed accordingly. We convert the measured syndrome components to sparse form via temporal difference encoding, which reduces the number of nonzero entries per column. For CSS codes, we can construct separate decoding matrices $D_X$ and $D_Z$ for $X$-type and $Z$-type errors respectively.

The circuit-level distance computation is then formulated as the following optimization problem. For $Z$-type errors:
\begin{equation}
\label{eq:dcirc}
d_{\text{circ}}^Z(\boldsymbol{\eta}) = \min_{\boldsymbol{\xi} \in \ker D_X} \sum_{j=1}^{M} \xi_j \quad \text{subject to} \quad \boldsymbol{\eta}^T \boldsymbol{\xi} = 1,
\end{equation}
where $\boldsymbol{\eta}$ is a vector in the span of rows of $D_X$ and rows of $D^L$ corresponding to logical $X$-type operators, but not in the row space of $D_X$ alone. This constraint ensures that $\boldsymbol{\eta}^T \boldsymbol{\xi} = 1$ is satisfiable for some $\boldsymbol{\xi} \in \ker D_X$ and selects logical errors. Since the exact decoding problem is NP-hard, we solve Eq.~\eqref{eq:dcirc} heuristically using the Belief Propagation with Ordered Statistics Decoding (BP-OSD) algorithm~\cite{Panteleev2021,bravyi2024high}. Specifically, we form the augmented parity check matrix $[D_X; \boldsymbol{\eta}^T]$ with syndrome $(0, \ldots, 0, 1)^T$, and the Hamming weight of the BP-OSD solution provides an upper bound $d_{\text{circ}}^{Z,\text{BP}}(\boldsymbol{\eta}) \geq d_{\text{circ}}^Z$.

A crucial property is that for a uniformly random choice of $\boldsymbol{\eta}$, the probability that $d_{\text{circ}}^{Z,\text{BP}}(\boldsymbol{\eta})$ equals the true $d_{\text{circ}}^Z$ is at least $1/2$ (whenever BP-OSD finds the optimal solution and $\boldsymbol{\eta}$ anticommutes with a minimum-weight logical operator). Therefore, by sampling $T$ independent random vectors $\boldsymbol{\eta}^1, \ldots, \boldsymbol{\eta}^T$ and taking the minimum, we obtain a systematically improvable upper bound:
\begin{equation}
d_{\text{circ}}^{Z,\text{BP}} := \min_{a=1,\ldots,T} d_{\text{circ}}^{Z,\text{BP}}(\boldsymbol{\eta}^a) \geq d_{\text{circ}}^Z.
\end{equation}
We compute $d_{\text{circ}}^{X,\text{BP}}$ analogously using $D_Z$ and logical $Z$-type rows of $D^L$, and report $d_{\text{circ}}^{\text{BP}} = \min(d_{\text{circ}}^{X,\text{BP}}, d_{\text{circ}}^{Z,\text{BP}})$ as the upper bound on the circuit-level distance.

\section{Numerical simulations of circuit-level performance}
\label{simulation}
\subsection{Circuit-level simulation and decoding}

All circuit-level simulations are performed using the open-source tool Stim~\cite{Gidney2021stimfaststabilizer}, a high-performance stabilizer-circuit simulator designed for quantum error correction.  We adopt the standard circuit-level depolarizing noise model.  Two-qubit depolarizing channels are applied after every two-qubit gate, and single-qubit depolarizing channels are applied after every idle operation. We ignored the single-qubit gate noise because it is typically very small compared to the two-qubit gate noise. Note that this assumption is also equivalent to compiling the circuit using CXSWAP gates, in which single-qubit gates do not appear. Concretely, single-qubit and two-qubit depolarizing channels are
\begin{equation}
\begin{aligned}
\mathcal{E}_{1}(\rho_1) &= (1-p)\rho_1 + \frac{p}{3}\sum_{P\in\{X,Y,Z\}} P \rho_1 P,\\
\mathcal{E}_{2}(\rho_2) &= (1-p)\rho_2 \\ &+ \frac{p}{15}
\sum_{\substack{P_1,P_2\in\{I,X,Y,Z\},\\P_1\otimes P_2\neq I \otimes I}} (P_1\otimes P_2)\,\rho_2\,(P_1\otimes P_2),
\end{aligned}
\end{equation}
where $\rho_1$ and $\rho_2$ denote single-qubit and two-qubit density matrices, respectively.  In addition, measurement readout and state initialization each suffer a bit-flip with probability $p$.

For decoding we employ the BP-OSD (Belief Propagation with Ordered Statistics Decoding) algorithm, using the open-source \textit{stimbposd}.  BP-OSD is a general-purpose decoder for quantum LDPC codes that performs well across a broad range of code families, although it is not guaranteed to be accurate for any particular code.  Crucially, we apply the identical decoder and identical decoding parameters to all codes compared in this work, including routing codes, BB codes, and rotated surface codes, ensuring a fair comparison.  The reported logical error rates should be regarded as a consistent lower bound on the achievable performance, not the ultimate limit. Both routing codes and BB codes can benefit from more accurate decoders. We note that all results are subject to improvement through the use of more specialized decoders in the future.

\subsection{Logical error rate and physical qubit overhead}

The logical error rate is estimated by simulating memory experiments, where half of the samples are initialized in the X basis to test the logical Z error rate, and the other half in the Z basis to test the logical X error rate. The logical error rate \( p_L \) is estimated as the average of the logical Z and logical X error rates. 

For a code with parameters $[[n,k,d]]$, let $P_{L,N}$ be the cumulative logical error rate accumulated over $N$ error-correction cycles.  The logical error rate per logical qubit per cycle is well approximated by
\begin{equation}
p_L = \frac{1}{k}\Bigl(1 - (1-P_{L,N})^{1/N}\Bigr) \approx \frac{P_{L,N}}{kN}.
\end{equation}
Throughout the numerical simulation in this work we fix $N = 10$ QEC cycles. For each data point in the logical error rate curve, we sampled at least 10,000 logical error events, so the standard deviation of \( p_L \) is no greater than \( 0.01 p_L \).

To quantify the practical advantage of routing codes, we compare their physical qubit overhead against that of the rotated surface code.  At the same physical error rate $p = 10^{-3}$, we simulate rotated surface codes with distances $d$ ranging from $3$ to $13$, each requiring $2d^{2}-1$ physical qubits per logical qubit.  For each routing code, we identify the  surface code distance $d^{\star}$ whose per-round logical error rate is closest to that of the routing code.  The ratio of physical qubit costs,
\begin{equation}
n_{\text{surface}}/n = \frac{2(d^{\star})^{2} - 1}{2 n/k},
\end{equation}
then directly measures the qubit savings offered by routing codes relative to the surface code.

\section{Hardware complexity analysis}
\label{sec:hardware_complexity}

\subsection{Hardware complexity on superconducting platforms}

We quantify the hardware resource overhead of a superconducting chip layout using the composite metric $C_{\text{hw}}$ introduced in Ref.~\cite{Mathews2026}. The metric aggregates four key quantities that dominate fabrication complexity in multi-layer architectures: the number of routing tiers ({Tier}), the average coupler length on higher tiers in units of the shortest qubit spacing ({Length}), the maximum average number of bump bonds per edge across all tiers ({Bump}), and the average number of through-silicon vias per edge on higher tiers ({TSV}).

Each quantity $q_i$ is linearly scaled between a baseline value $b_i$, representing state-of-the-art surface-code hardware, and an optimistic value $p_i$ reflecting near-future attainable technology. The scaled quantity is $c_i = (q_i - b_i)/(p_i - b_i)$, and the composite score is
\begin{equation}
C_{\mathrm{hw}} = 1 + \frac{1}{4} \sum_i c_i.
\end{equation}
By construction, the surface code yields $C_{\mathrm{hw}} = 1$ as a baseline, and a design achieving all four optimistic targets yields $C_{\mathrm{hw}} = 2$.  The optimistic values were set to 5 tiers, 10 times longer long-range couplers than short range couplers, 4 bump transitions, and 3 TSVs per coupler. In the main text, we report $C_{\mathrm{hw}}$ for each routing code and BB code as evaluated by the HAL framework under its default parameters, taking the best layout among the available placement options.

\subsection{Minimum scheduling depth for atom movement}

To determine the number of time steps required for executing one layer of long-range two-qubit gates on a neutral-atom platform, we solve a minimum movement depth problem on a toroidal grid. We consider an $l \times m$ grid where atoms at the X and Z syndrome qubit positions are displaced by their respective routing vectors to reach the neighboring data qubits. The trajectory of each movement follows the corresponding coupling vector. All movements are straight segments under periodic boundary conditions. The objective is to partition these movements into the fewest time steps such that trajectories within the same step do not intersect in a plane, including endpoint contacts.

The problem is reduced to coloring a conflict graph $G$ whose vertices are the atom movements. Two vertices are adjacent if their toroidal trajectories intersect, which is verified by testing planar segment intersections between one segment and the integer translates $(km,\ell n)$ of the other. The minimum number of time steps equals the chromatic number $\chi(G)$. A breadth-first search checks whether $G$ is bipartite. If $G$ has no edges, $\chi(G)=1$ (one time step suffices). If $G$ is bipartite and has at least one edge, then $\chi(G)=2$. If $G$ is not bipartite, a greedy coloring provides an upper bound, and an exact backtracking search finds the optimal $\chi(G)$.

Using this scheduling model, we evaluated the time step requirements for both routing codes and BB codes. For routing codes, any layer of two-qubit gates necessitates a partition into two time steps to resolve conflicts arising from connections that traverse the toroidal boundaries. In contrast, for BB codes, the concurrent execution of coupling layers corresponding to vectors $(6,3)$ and $(3,6)$ requires four time steps. It is important to note that within the standard 7-layer error correction circuit for BB codes, the parallel execution of $(6,3)$ and $(3,6)$ couplings is structurally unavoidable.

\section{More details of code instances}
\label{sec:more_instances}

In this section we provide additional details and examples of routing codes that complement the main text. Table~\ref{tab:S1} lists the routing vector sequences and torus dimensions that uniquely determine the codes presented in Table~\ref{tab:selected_codes} of the main text. All other code properties, including parity-check matrices and hardware connectivity, can be independently reconstructed from these parameters.

Tables~\ref{tab:addtional_codes} present information on additional weight-7 routing codes as well as routing codes with larger weights. These codes are carefully selected representatives from our search that lead in encoding rate. We also require that $d_{\mathrm{circ}}$ does not exhibit a significant degradation, which comes from hook errors. As the weight increases to 11, we find codes with very high encoding rates, significantly surpassing those of surface codes or BB codes.  In fact, if the circuit-level distance test is removed, we have found even more codes with higher encoding rates. However, we do not present them here, as they could be misleading regarding practical performance. 

We emphasize that the instances reported here represent only a fraction of the routing code landscape. Additional codes can be obtained by varying the search parameters. We also note an empirical observation. Among these good-performing routing codes, the vector sequences consistently satisfy \(\mathbf{v}_t = \mathbf{w}_{T-t+1}\), meaning the X and Z routing sequences are time-reversed versions of each other. In fact, we also constructed several other special symmetry conditions, but none yielded better results than time-reversal symmetry. While a complete theoretical understanding of why this particular symmetry is so prevalent among the good codes remains an open question, we report it here as a useful condition for future searches.

\begin{table*}[htpb]
\centering
\renewcommand{\arraystretch}{1.15}
\small
\begin{tabular}{c c c c}
\toprule
$[[n,k,d]]$ & $\qquad$ $l, m$  $\qquad$& $\qquad$ $\qquad$   $\textbf{Routing vector sequence}$ \(\mathbf{\{v_t\}} \) $\qquad$ $\qquad$ &$\qquad$ $\qquad$ $\textbf{Routing vector sequence}$ \(\mathbf{\{w_t\}} \)  $\qquad$ $\qquad$ \\
\midrule
$[[54,8,6]]$   & $18, 6$  & [(0, 1), (1, 0), (3, 0), (0, 1), (3, 0), (1, 0), (0, 1)] & [(0, 1), (1, 0), (3, 0), (0, 1), (3, 0), (1, 0), (0, 1)] \\
$[[70,8,7]]$   & $14, 10$ & [(1, 0), (2, 1), (0, 1), (0, 1), (0, 1), (2, 1), (1, 0)] & [(1, 0), (2, 1), (0, 1), (0, 1), (0, 1), (2, 1), (1, 0)] \\
$[[80,8,8]]$   & $16, 10$ & [(1, 0), (0, 1), (0, 1), (0, 1), (4, 1), (4, 1), (1, 0)] & [(1, 0), (4, 1), (4, 1), (0, 1), (0, 1), (0, 1), (1, 0)] \\
$[[90,8,9]]$   & $18, 10$ & [(1, 0), (0, 1), (0, 1), (0, 1), (2, 1), (2, 1), (1, 0)] & [(1, 0), (2, 1), (2, 1), (0, 1), (0, 1), (0, 1), (1, 0)] \\
$[[100,8,10]]$ & $20, 10$ & [(1, 0), (0, 1), (0, 1), (2, 1), (2, 1), (2, 1), (1, 0)] & [(1, 0), (2, 1), (2, 1), (2, 1), (0, 1), (0, 1), (1, 0)] \\
$[[110,8,11]]$ & $22, 10$ & [(1, 0), (0, 1), (0, 1), (6, 1), (6, 1), (0, 1), (1, 0)] & [(1, 0), (0, 1), (6, 1), (6, 1), (0, 1), (0, 1), (1, 0)] \\
$[[140,8,12]]$ & $28, 10$ & [(1, 0), (0, 1), (0, 1), (4, 1), (0, 1), (0, 1), (1, 0)] & [(1, 0), (0, 1), (0, 1), (4, 1), (0, 1), (0, 1), (1, 0)] \\
$[[140,8,13]]$ & $28, 10$ & [(1, 0), (6, 1), (6, 1), (0, 1), (0, 1), (0, 1), (1, 0)] & [(1, 0), (0, 1), (0, 1), (0, 1), (6, 1), (6, 1), (1, 0)] \\
\bottomrule
\end{tabular}
\caption{Generating parameters of the routing codes presented in Table~1 of the main text. For each code, the torus dimensions $l,m$ and the routing vector sequences $\{\mathbf{v}_t\}$ and $\{\mathbf{w}_t\}$ are listed. The second sequence is the time-reversed version of the first. These parameters uniquely determine all code properties.}
\label{tab:S1}
\end{table*}

\begin{table*}[htbp]
\centering
\renewcommand{\arraystretch}{1.15}
\small
\begin{tabular}{c c c c c c}
\toprule
\textbf{Code} & $[[n,k,d]]$ & $\,$ $d_{\text{circ}}$  $\,$&  \textbf{Encoding rate} & $\quad$  $l, m$  $\quad$ & $\textbf{Routing vector sequence}$ \(\mathbf{\{v_t\}} \)\\
\midrule
\multirow{2}{*}{\shortstack{Routing codes\\weight-7}} 
& $[[240,8,17]]$  & $\leq 17$ & $1 / 30.00$ & 16, 30 & [(1, 0), (2, 1), (2, 1), (0, 1), (0, 1), (2, 1), (1, 0)]\\
& $[[260,8,18]]$  & $\leq 17$ & $1 / 32.50$ & 52, 10 & [(1, 0), (4, 1), (4, 1), (0, 1), (0, 1), (0, 1), (1, 0)]\\
\midrule
\multirow{6}{*}{\shortstack{Routing codes\\weight-8}} 
& $[[48,10,6]]$  & $\leq 6$ & $1 / 4.80$ & 8, 12 & [(5, 2), (0, 1), (0, 1), (5, 2), (1, 0), (0, 1), (0, 1), (1, 0)]\\
& $[[128,18,8]]$  & $\leq 8$ & $1 / 7.11$ & 32, 8 & [(0, 1), (5, 0), (5, 0), (0, 1), (0, 1), (1, 0), (1, 0), (0, 1)]\\
& $[[192,16,12]]$  & $\leq 12$ & $1 / 12.00$ & 16, 24 & [(5, 2), (0, 1), (0, 1), (5, 2), (1, 0), (0, 1), (0, 1), (1, 0)]\\
& $[[204,10,15]]$  & $\leq 15$ & $1 / 20.40$ & 12, 34 & [(0, 1), (5, 2), (5, 2), (1, 0), (1, 0), (1, 0), (5, 2), (0, 1)]\\
& $[[384,18,16]]$  & $\leq 16$ & $1 / 21.33$ & 32, 24 & [(1, 0), (0, 1), (0, 1), (6, 1), (6, 1), (0, 1), (0, 1), (1, 0)]\\
& $[[228,10,17]]$  & $\leq 16$ & $1 / 22.80$ & 12, 38 & [(0, 1), (5, 2), (1, 0), (1, 0), (1, 0), (5, 2), (5, 2), (0, 1)]\\
\midrule
\multirow{6}{*}{\shortstack{Routing codes\\weight-9}} 
& $[[72,16,6]]$  & $\leq 6$ & $1 / 4.50$ & 12, 12 & [(3, 0), (0, 1), (1, 0), (0, 1), (0, 1), (0, 1), (1, 0), (0, 1), (3, 0)]\\
& $[[200,32,8]]$  & $\leq 8$ & $1 / 6.25$ & 20, 20 & [(1, 0), (3, 0), (0, 1), (0, 1), (1, 0), (0, 1), (0, 1), (3, 0), (1, 0)]\\
& $[[108,12,9]]$  & $\leq 9$ & $1 / 9.00$ & 12, 18 & [(5, 2), (0, 1), (1, 0), (0, 1), (0, 1), (0, 1), (1, 0), (0, 1), (5, 2)]\\
& $[[110,8,11]]$  & $\leq 11$ & $1 / 13.75$ & 10, 22 & [(3, 2), (0, 1), (0, 1), (1, 0), (1, 0), (1, 0), (0, 1), (0, 1), (3, 2)]\\
& $[[250,16,15]]$  & $\leq 14$ & $1 / 15.62$ & 50, 10 & [(1, 0), (1, 0), (6, 1), (0, 1), (1, 0), (0, 1), (6, 1), (1, 0), (1, 0)]\\
& $[[360,16,18]]$  & $\leq 18$ & $1 / 22.50$ & 36, 20 & [(5, 2), (0, 1), (1, 0), (0, 1), (0, 1), (0, 1), (1, 0), (0, 1), (5, 2)]\\
\midrule
\multirow{6}{*}{\shortstack{Routing codes\\weight-10}} 
& $[[98,26,5]]$  & $\leq 5$ & $1 / 3.77$ & 14, 14 & [(1, 0), (0, 1), (2, 1), (2, 1), (1, 0), (1, 0), (2, 1), (2, 1), (0, 1), (1, 0)]\\
& $[[70,10,7]]$  & $\leq 7$ & $1 / 7.00$ & 10, 14 & [(1, 0), (1, 0), (0, 1), (0, 1), (3, 2), (3, 2), (0, 1), (0, 1), (1, 0), (3, 2)]\\
& $[[288,30,10]]$  & $\leq 10$ & $1 / 9.60$ & 24, 24 & [(0, 1), (1, 0), (3, 2), (0, 1), (0, 1), (0, 1), (0, 1), (3, 2), (1, 0), (0, 1)]\\
& $[[192,12,14]]$  & $\leq 14$ & $1 / 16.00$ & 32, 12 & [(0, 1), (1, 0), (0, 1), (1, 0), (0, 1), (0, 1), (1, 0), (2, 1), (1, 0), (2, 1)]\\
& $[[190,10,16]]$  & $\leq 16$ & $1 / 19.00$ & 38, 10 & [(1, 0), (1, 0), (0, 1), (0, 1), (4, 3), (4, 3), (0, 1), (4, 3), (1, 0), (1, 0)]\\
& $[[336,16,18]]$  & $\leq 17$ & $1 / 21.00$ & 24, 28 & [(1, 0), (1, 0), (0, 1), (1, 0), (0, 1), (0, 1), (5, 2), (0, 1), (5, 2), (5, 2)]\\
\midrule
\multirow{6}{*}{\shortstack{Routing codes\\weight-11}} 
& $[[50,16,5]]$  & $\leq 5$ & $1 / 3.12$ & 10, 10 & [(1, 0), (0, 1), (1, 0), (3, 0), (0, 1), (0, 1), (0, 1), (3, 0), (1, 0), (0, 1), (1, 0)]\\
& $[[100,16,8]]$  & $\leq 8$ & $1 / 6.25$ & 10, 20 & [(3, 0), (0, 1), (1, 0), (0, 1), (0, 1), (3, 0), (0, 1), (0, 1), (1, 0), (0, 1), (3, 0)]\\
& $[[90,12,9]]$  & $\leq 9$ & $1 / 7.50$ & 18, 10 & [(5, 0), (0, 1), (1, 0), (1, 0), (0, 1), (0, 1), (0, 1), (1, 0), (5, 0), (0, 1), (5, 0)]\\
& $[[150,20,12]]$  & $\leq 11$ & $1 / 7.50$ & 30, 10 & [(1, 0), (0, 1), (1, 0), (1, 0), (6, 1), (6, 1), (6, 1), (1, 0), (1, 0), (0, 1), (1, 0)]\\
& $[[200,24,14]]$  & $\leq 14$ & $1 / 8.33$ & 20, 20 & [(1, 0), (0, 1), (1, 0), (1, 0), (6, 1), (6, 1), (6, 1), (1, 0), (1, 0), (0, 1), (1, 0)]\\
& $[[200,16,17]]$  & $\leq 17$ & $1 / 12.50$ & 40, 10 & [(1, 0), (0, 1), (1, 0), (1, 0), (6, 1), (6, 1), (6, 1), (1, 0), (1, 0), (0, 1), (1, 0)]\\
\bottomrule
\end{tabular}
\caption{Additional details on the constructed routing codes. The sequence $\{\mathbf{w}_t\}$ is omitted since it is the time-reversed version of $\{\mathbf{v}_t\}$. As the weight increases, codes with significantly higher encoding rates are obtained, surpassing the performance of BB codes.}
\label{tab:addtional_codes}
\end{table*}

\end{appendix}

\clearpage

\bibliographystyle{apsrev4-2}
\bibliography{ref}

\end{document}